%% file: main.tex
\documentclass[11pt]{article}
\usepackage[english]{babel}
\usepackage[numbers,sort&compress]{natbib}
\usepackage{graphicx}
\usepackage[dvipsnames,table,xcdraw]{xcolor}
\usepackage{framed}
\usepackage[normalem]{ulem}
\usepackage{amsmath}
\usepackage{amsthm,thmtools,mathtools}
\usepackage{amssymb}
\usepackage{bbm}
\usepackage{amsfonts}
\usepackage{resizegather}
\usepackage{enumerate}
\usepackage{algpseudocode}
\usepackage[linesnumbered,ruled,vlined,noend]{algorithm2e}
\SetKwRepeat{Do}{do}{while}
\usepackage{nomencl}
\usepackage{enumitem}
\usepackage{wrapfig} 
\usepackage{caption}
\usepackage{subcaption}
\usepackage[margin=1in]{geometry}
\usepackage[
  pagebackref,
  colorlinks=true,
  urlcolor=blue,
  linkcolor=black,
  citecolor=blue,
]{hyperref}
\usepackage[nameinlink]{cleveref}
\theoremstyle{definition}
\newtheorem{theorem}{Theorem}
\newtheorem{observation}{Observation}
\newtheorem{algo}{Algorithm}
\newtheorem{corollary}{Corollary}
\newtheorem{lemma}{Lemma}

\newtheorem{definition}{Definition}
\newtheorem{example}{Example}

\newtheorem*{remark}{Remark}

\newenvironment{numberedtheorem}[1]{%
\renewcommand{\thetheorem}{#1}%
\begin{theorem}}{\end{theorem}\addtocounter{theorem}{-1}}

\newenvironment{numberedlemma}[1]{%
\renewcommand{\thelemma}{#1}%
\begin{lemma}}{\end{lemma}\addtocounter{lemma}{-1}}

\newenvironment{numberedcorollary}[1]{%
\renewcommand{\thecorollary}{#1}%
\begin{corollary}}{\end{corollary}\addtocounter{corollary}{-1}}

\DeclareMathOperator{\E}{\mathbb{E}}
\DeclareMathOperator{\sw}{\text{SW}}

\DeclareMathOperator{\T}{\mathcal{T}}

\DeclareMathOperator{\opt}{OPT}
\DeclareMathOperator{\PDim}{PDim}

\DeclarePairedDelimiter\floor{\lfloor}{\rfloor}
\renewcommand{\vec}[1]{\mathbf{#1}}
\newcommand{\eps}{\varepsilon}

\usepackage{soul}

\newcommand{\Tau}{\mathcal{T}}
\SetCommentSty{mycommfont}

\makeatletter
\newcounter{modules}

\makeatother

\title{Setting Fair Incentives to Maximize Improvement}
\author{}
\author{Saba Ahmadi\thanks{Toyota Technological Institute at Chicago. Email: \texttt{saba@ttic.edu}. Author was supported by the National Science Foundation grant CCF-1733556, and the Simons Foundation under the Simons Collaboration on the Theory of Algorithmic Fairness.} \and%
Hedyeh Beyhaghi\thanks{Carnegie Mellon University. Email: \texttt{hhedyeh@cmu.edu}. This work was done while the author was a Postdoctoral Researcher at Toyota Technological Institute at Chicago.
} \and%
Avrim Blum\thanks{Toyota Technological Institute at Chicago. Email: \texttt{avrim@ttic.edu}. This work was supported in part by the National Science Foundation under grants CCF-1815011 and CCF-1733556, and the Simons Foundation under the Simons Collaboration on the Theory of Algorithmic Fairness.} \and%
Keziah Naggita\thanks{Toyota Technological Institute at Chicago. Email: \texttt{knaggita@ttic.edu}. This work was supported in part by the National Science Foundation under grant CCF-1815011, and the Simons Foundation under the Simons Collaboration on the Theory of Algorithmic Fairness.}}
\date{}

\usepackage[parfill]{parskip}
\begin{document}
\sloppy

\maketitle

\begin{abstract}
\input{abstract}
\end{abstract}

\newpage
\input{intro}
\input{model}

\input{max_total_imp}

\input{max_min}

\input{approx}

\input{learning}
\input{extensions}


\bibliographystyle{plainnat}
\newpage
{\footnotesize
\bibliography{ref}
}


\newpage
\appendix
\input{appendix-missing-proof}
\input{FPTAS-maxmin}
\input{appendix_missing_approximation}

\input{appendix_example}

\end{document}

%% file: abstract.tex
We consider the problem of helping agents improve by setting short-term goals. Given a set of target skill levels, we assume each agent will try to improve from their initial skill level to the closest target level within reach (or do nothing if no target level is within reach). We consider two models: the \emph{common} improvement capacity model, where agents have the same limit on how much they can improve, and the \emph{individualized} improvement capacity model, where agents have individualized limits. Our goal is to optimize the target levels for social welfare and fairness objectives, where \emph{social welfare} is defined as the total amount of improvement, and fairness objectives are considered where the agents belong to different underlying populations. We prove algorithmic, learning, and structural results for each model.

A key technical challenge of this problem is the non-monotonicity of social welfare in the set of target levels, i.e., adding a new target level may decrease the total amount of improvement; agents who previously tried hard to reach a distant target now have a closer target to reach and hence improve less. This especially presents a challenge when considering multiple groups because optimizing target levels in isolation for each group and outputting the union may result in arbitrarily low improvement for a group, failing the fairness objective. Considering these properties, we provide algorithms for optimal and near-optimal improvement for both social welfare and fairness objectives. These algorithmic results work for both the common and individualized improvement capacity models. Furthermore, despite the non-monotonicity property and interference of the target levels, we show a placement of target levels exists that is approximately optimal for the social welfare of each group. Unlike the algorithmic results, this structural statement only holds in the common improvement capacity model, and we illustrate counterexamples to this result in the individualized improvement capacity model. Finally, we extend our algorithms to learning settings where we have only sample access to the initial skill levels of agents.

%% file: intro.tex
\section{Introduction}

Consider a vocational school designed to improve participants' skills and help prepare them for the workforce. The participating students have different skill levels that the school has access to by a pre-screening method. In order to accommodate different skill levels, the organizer designs multiple projects at different difficulty levels. Succeeding in completing a project has the effect of causing  students to improve their skills to that project level. The students only get credit for projects above their initial level, and each student is assumed to pick the closest difficulty level above their initial skill that is within reach. If students feel all the projects are out of reach, they get discouraged and do not participate. The designer's goal is to maximize the total improvement both with and without fairness considerations.

Mathematically, we formulate this problem as follows. There are $n$ agents belonging to $g$ distinct groups. Agent $i$ has an initial skill level, $p_i \in \mathbb{Z}_{\geq 0}$, and can increase their skill by at most $\Delta_i$ which is called the ``improvement capacity''. Given a set of target levels $\T \subset \mathbb{Z}_{\geq 0}$, agent $i$ improves to the closest target $\tau \in \T$ such that $\tau > p_i$ and $\tau \leq p_i+\Delta_i$ if such target exists; otherwise it stays at $p_i$. This model also captures scenarios such as designing promotion levels in firms, and more broadly designing incentives for self-improvement to optimize efficiency and fairness.

This problem formulation gives rise to multiple challenges. First, optimizing improvement for a set of agents may conflict with another set. Consider a beginner-level agent (skill level $B$) and an intermediate-level (skill level $I$). Agent $I$ finds any level up to $\tau_I$ within reach. Therefore, we need to design a project at level $\tau_I$ for this agent to improve maximally. On the other hand, $B$ has the capacity to improve until $\tau_B$, where $I < \tau_B < \tau_I$ --- See \Cref{fig:interfere}. Now, consider both target levels $\tau_B$ and $\tau_I$. Since agent $I$ now has a closer target of $\tau_B$, this agent no longer achieves its maximum improvement, {and only reaches skill level $\tau_B$.} 
Secondly, there is non-monotonicity in the placement of target levels, i.e., adding a new target to the current placement may decrease the total amount of improvement. Consider a beginner-level ($B$) and an intermediate-level ($I$) agent and a target, $\tau$, achievable by both agents --- See \Cref{fig:nonmonotone_sub}. Designing a new project at level $\tau'$ between $B$ and $\tau$ decreases the total amount of improvement since one agent (if $B <\tau'\leq I$) or both agents (if $I<\tau'<\tau$) switch from improving to $\tau$ to improving to $\tau'$, which requires less improvement.

\begin{figure}[ht!]
    \centering
    \begin{subfigure}[b]{0.46\textwidth}
        \centering
        \includegraphics[height=1cm]{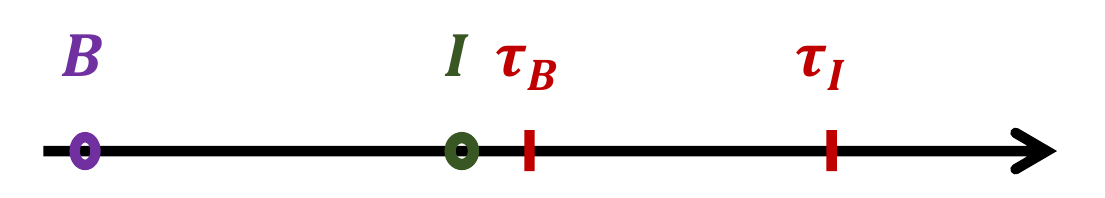}
        \caption{Conflict in optimizing improvement.}
        \label{fig:interfere}
    \end{subfigure}
    \hfill
    \begin{subfigure}[b]{0.46\textwidth}
        \centering
        \includegraphics[height=1cm]{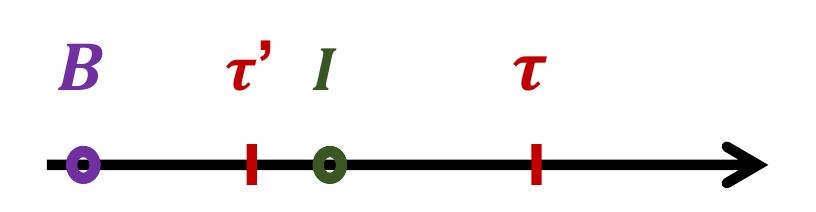}
        \caption{Non-monotonicity in set of target levels.}
        \label{fig:nonmonotone_sub}
    \end{subfigure}
    \caption{Challenges in designing optimal target levels.}
    \label{fig:conflict_and_nonmonotone}
\end{figure}

\textbf{Main Results.} In this work, we consider algorithmic, fairness, and learning-theoretic formulations, where a set of optimal target levels must be found in the presence of effort-bounded agents. {We use \textit{social welfare} as the notion of efficiency and define it as the total amount of improvement. Also, we define \textit{social welfare for a given group} as the amount of improvement that group achieves.} We consider two models: (1) the common improvement capacity model, where agents have the same limit $\Delta$ on how much they can improve, and (2) the individualized improvement capacity model, where agents have individualized limits $\Delta_i$. 

The main results of the paper are:
\begin{enumerate}
    \item An efficient algorithm for placement of target levels to maximize social welfare. (\Cref{sec:max_total_improvement})
    \item An efficient algorithm for outputting the Pareto-optimal outcome for the social welfare of multiple groups. In particular, this can output the max-min fair solution that maximizes the minimum total improvement across groups. (\Cref{sec:max_min}) 
    \item A structural result on Pareto-optimal solutions: there exists a placement of target levels that simultaneously is approximately optimal for each group. {More explicitly,  when there are a constant number of groups, the total improvement for each group is a constant-factor approximation of the maximum  improvement that we could provide that group if it were the only group under consideration. \textbf{This is our main contribution.}} (\Cref{sec:approx_optimality})     
    \item An efficient learning algorithm for near-optimal placement of target levels. (\Cref{sec:gen_guarantees})
\end{enumerate}

The algorithmic results work for both the common and individualized improvement capacity models. However, the structural result only holds in the common improvement capacity model, and we illustrate examples where achieving any nontrivial fraction of optimal for all groups is not possible in the individualized capacity model.


\paragraph*{Related work.}
Our work broadly falls under two general research areas: social welfare maximization in mechanism design and algorithmic fairness. Specifically, the closest topics to our paper are designing portfolios for consumers to minimize loss of returns \cite{Diana-etal}, designing badges to steer users' behavior \cite{steering-user-beha}, and the literature on strategic classification.

Closest to our work is \citet{Diana-etal} who consider a model where each agent has a risk tolerance, observed as a real number, and must be assigned to a portfolio with risk lower than what they can tolerate. The goal of the mechanism designer is to design a small number of portfolios that minimizes the sum of the differences between the agent's risk tolerance and the risk of the portfolio they take; in other words, it minimizes the loss of returns. Since this is a minimization problem where each agent selects the closest target (portfolio) below their risk tolerance, adding any new target can only help with the objective function. Therefore, unlike our model, there is no conflict between targets, and the objective function is monotone in the set of targets.

Designing targets to incentivize agents to take specific actions is also a common feature of online communities and social media sites. In these platforms, there is a mechanism for rewarding user achievements based on a system of \emph{badges} (similar to targets in our model) \cite{steering-user-beha, gamification, babaioff2012bitcoin, burke2009feed, burke2011plugged}. Among such papers, the closest to ours is \citet{steering-user-beha} who investigate how to optimally place badges in order to induce particular user behaviors, among other things. They consider a dynamic setting with a single user type interested in a particular distribution of actions and a mechanism designer whose objective is to set badges to motivate a different distribution of actions. Compared to our work, their model is more general in the sense that users can spend effort on different actions (improve in multiple dimensions), but also more specific, in the sense that there is only one user type; therefore, unlike our model there is no conflict between different users and adding more badges for the desired action always helps with steering the users in that direction (it is a monotone setting).

Another line of work that is relevant is strategic classification. In most cases, agents are fraudulently strategic, that is to say, game the decision-making model to get desired outcomes (see \cite{Hardt2016,revealed_preferences,Hu:2019:,Milli2018TheSC, Ahmadi2021TheSP,adversarial-games-pred,Frankel2019ImprovingIF,braverman2020role} among others). In other cases, in addition to actions only involving gaming the system, agents can also perform actions that truthfully change themselves to become truly qualified (see \cite{Kleinberg2018HowDC, harris2021stateful,Alon2020MultiagentEM,xiao2020optimal,Miller2019StrategicCI,Haghtalab2020MaximizingWW, Bechavod2020CausalFD,Shavit2020LearningFS} among others). In this paper, we assume agents only truthfully change themselves and, therefore, focus on incentivizing agents to improve as much as they can.


\paragraph*{Organization of the Paper.}
{\Cref{sec:model} formally introduces the general model settings and definitions used in the paper, and \Cref{sec:max_total_improvement} provides an efficient algorithm for the problem of maximizing total improvement.
In \Cref{sec:max_min}, we provide algorithms that output Pareto optimal solutions for groups' social welfare, including a solution that maximizes the minimum improvement per group.
In \Cref{sec:approx_optimality}, we provide an algorithm that finds the best simultaneously approximately optimal improvement per group and show it provides a constant approximation when the number of groups is constant.
In \Cref{sec:gen_guarantees}, we provide efficient learning algorithms which generalize the previous results to a setting where there is only sample access to agents, and \Cref{sec:extensions} provides further extensions to our main problems.
All missing proofs are deferred to the appendix.}

%% file: model.tex
\section{Model and Preliminaries}\label{sec:model}

There are $n$ agents $1, \ldots, n$. Agent $i$ is associated with two quantifiers: initial skill level, $p_i$, and \textit{improvement capacity}, $\Delta_i$, which determines the maximum amount agent $i$ can improve its skill. For the majority of the paper, we assume $p_i$ and $\Delta_i$ belong to $\mathbb{Z}_{\geq 0}$; however, some of our results hold more generally for real numbers.\footnote{All our examples that do not use integer numbers can be converted to integer numbers by scaling.}

We consider two different models. The \textit{common} and the \textit{individualized} improvement capacity models. In the first model, all agents have the same improvement capacity, i.e., $\Delta_i$ are equal across agents; we substitute $\Delta_i$ with $\Delta$ in this case. The second model is a generalization where $\Delta_i$ may have different values. We use $\Delta_{max} = \max\{\Delta_1,\cdots,\Delta_n\}$.

Our solution is a finite set of target levels $\T \subset \mathbb{Z}_{\geq 0}$. We assume we are given a maximum number of allowed target levels $k$ (if $k=n$, this is equivalent to allowing an unbounded number of target levels).

\textbf{Agents behavior.} Given target levels $\T \subset \mathbb{Z}_{\geq 0}$, agent $i$ aims for the closest target above its initial skill if it can reach to that target given its improvement capacity. More formally, agent $i$ aims for $\min \{\tau \in \T : p_i<\tau\leq p_i+\Delta_i\}$ if such $\tau$ exists and improves from $p_i$ to $\tau$. If no such target exists, agent $i$ does not improve and its final skill level remains the same as the initial skill level $p_i$. 

We use \textit{social welfare} ($\sw$) as our notion of efficiency and define it as the total amount of improvement of agents. 


\textbf{Groups and fairness notion.} Each agent belongs to one of $g$ distinct groups {$G_1,\cdots, G_g$}. Given any set of target levels, the social welfare of group $\ell$, $\sw_\ell$, is defined as the total amount of improvement for agents in that group.\footnote{Although the results are presented for the \emph{total} improvement objective, they also hold for the \emph{average} improvement objective.}
We are interested in Pareto-optimal solutions for groups' social welfare. A solution $\T$ is Pareto-optimal (is on the Pareto frontier) if there does not exist $\T'$ in which all groups gain at least as much social welfare, and one group gains strictly higher. In particular, the Pareto frontier includes the max-min solution that maximizes the minimum social welfare across groups.  In this paper, we focus on two natural fairness notions: one is the max-min solution described above, and the other is the notion of simultaneous approximate optimality given below.

\begin{definition}[Simultaneous $\alpha$-approximate optimality.] A solution with at most $k$ targets is simultaneously approximately optimal for each group with approximation factor $0 \leq \alpha \leq 1$ if, for each group $\ell$, the social welfare of group $\ell$ is at least an $\alpha$ fraction of the maximum social welfare achievable for group $\ell$ using at most $k$ targets.
\end{definition}

\subsection{Basic Properties of Optimal Target Sets}\label{sec:basics}

This section provides a simple structural result on optimal set of target levels. The following observation determines the potential positions of the targets in an optimal solution.

\begin{observation} \label{obs:potential_targets}
Without loss of optimality, the targets in an optimal solution are either at positions $p_i + \Delta_i$ or $p_i$ for some $i \in \{1, 2, \ldots, n\}$. Consider a solution where target $\tau$ does not satisfy this condition. By shifting $\tau$ to the right as long as it does not cross $p_i + \Delta_i$ or $p_i$ for any $i$, the total amount of improvement weakly increases: This transformation does not change the sets of agents that reach each target, and only increases the improvement of agents aiming for $\tau$.
\end{observation}

Observation \ref{obs:potential_targets} motivates the following definition. 

\begin{definition}[$\T_p$]
The set of potential optimal target levels, $\T_p$, is defined as $\bigcup_{i=1}^n\{p_i, p_i+\Delta_i\}$.
\end{definition}

%% file: max_total_imp.tex
\section{Maximizing Total Improvement }\label{sec:max_total_improvement}

In this section, we provide an efficient dynamic programming algorithm for finding a set of $k$ target levels that maximizes total improvement 
for a collection of $n$ agents. \Cref{def:recurrence-dp-one-group} 
{provides the details} of the dynamic programming algorithm. We bound its time-complexity in~\Cref{thm:total-improvement}.

In~\Cref{def:recurrence-dp-one-group}, the recursion function $T(\tau,\kappa)$ 
finds the best set of at most $\kappa$ target levels 
{for} agents on or to the right of $\tau$. 
Recall that any target $\tau$ only affects the agents on its left, and agent $i$ such that $p_i < \tau$ never selects $\tau' > \tau$ in presence of $\tau$. Utilizing these properties, the 
main idea for the recursive step (item $3$ in \Cref{def:recurrence-dp-one-group}) is to first consider the potential leftmost targets $\tau' > \tau$ and use the smaller subproblem of finding the optimal targets for agents on or to the right of $\tau'$ with one less available target level; i.e., $T(\tau', \kappa-1)$.  
To optimize over the potential leftmost target levels, $\tau'$, we first evaluate the performance of each potential target by improvement of agents who reach it; i.e., $i$ such that $\tau \leq p_i < \tau'$ and $\tau'-p_i\leq \Delta_i$, where agent $i$ improves by $\tau'-p_i$. Next, we add the performance of each potential leftmost target to the optimal improvement of the remaining subproblem and pick the leftmost target that maximizes this summation.

\begin{algo}
\label{def:recurrence-dp-one-group}

Run dynamic program based on function $T$, defined below, that takes $\cup_i \{p_i\}$ and $k$ as input and outputs $T(\tau_{\min},k)$, as the optimal improvement, and $S(\tau_{\min},k)$, as the optimal set of targets; where $\tau_{\min}= \min\{\tau \in \T_p\}$ and $\tau_{\max}= \max\{\tau \in \T_p\}$. $T(\tau,\kappa)$ captures the maximum improvement possible for agents on or to the right of $\tau\in \T_p$ when at most $\kappa$ target levels can be selected. Function $T$ is defined as follows.

\begin{itemize}

\item[1)]  For any $\tau\in \T_p$, 
$ \; T(\tau,0)=0$.

\item[2)]  For any $1\leq \kappa\leq k$, $ \; T(\tau_{\max},\kappa)=0$.

\item[3)]  For any $\tau\in \T_p, \tau< \tau_{\max}$ and $1\leq \kappa\leq k$:

\[T(\tau,\kappa) = \max_{\tau'\in \Tau_p \ \text{s.t} \ \tau' > \tau}\Bigg(T(\tau',\kappa-1) \ + \sum_{\tau\leq p_i <\tau'\text{ s.t. } \tau'-p_i \leq\Delta_i}(\tau'-p_i)\Bigg)\]
\end{itemize}

$S(\tau,\kappa)$ keeps track of the optimal set of targets corresponding to $T(\tau,\kappa)$.
\end{algo}

{The following theorem proves} the correctness of the dynamic programming algorithm and bounds its time-complexity.
\begin{theorem}
\Cref{def:recurrence-dp-one-group} finds a set of targets that achieves the optimal social welfare (maximum total improvement)
that is feasible using at most $k$ targets given $n$ agents. The algorithm runs in $\mathcal{O}(n^3)$.
\label{thm:total-improvement}
\end{theorem}
\begin{proof}
{See \Cref{app:missing_max_tot}.}
\end{proof}

%% file: max_min.tex
\section{Pareto Optimality and Maximizing Minimum Improvement}\label{sec:max_min}

In this section, we provide a dynamic programming algorithm that constructs the Pareto frontier for groups' social welfare. By iterating through all Pareto-optimal solutions, we can find the solution that maximizes minimum improvement across all groups in pseudo-polynomial time. Next, we provide a Fully Polynomial Time Approximation Scheme (FPTAS) for this objective. 

In~\Cref{recurrence-exact-fairness-objective}, we provide a dynamic program that constructs the Pareto frontier for groups' social welfare. 
In contrast to~\Cref{def:recurrence-dp-one-group} where the algorithm only needs to store an optimal solution for each subproblem, here for each subproblem the algorithm stores a set containing \emph{all} $g$-tuples of groups' improvements $(I_1, I_2,\cdots, I_g)$ that are simultaneously achievable for groups $\{G_1,\cdots, G_g\}$. Similar to the recurrence in~\Cref{def:recurrence-dp-one-group}, we consider the potential left-most targets $\tau'$ and subproblems for agents on or to the right of $\tau'$ with one less available target level; i.e., $T(\tau', \kappa-1)$. Particularly, in item $3$ of \Cref{recurrence-exact-fairness-objective}, we consider all combinations of potential left-most targets $\tau'$ and their corresponding subproblems. To evaluate the performance, for any potential leftmost target $\tau’ \in \T_p$ and $\tau’ > \tau$, we compute the improvement of all agents 
reaching to $\tau’$ from each group separately, i.e., $i \in G_\ell$ such that $\tau \leq p_i < \tau'$ and $\tau'-p_i\leq \Delta_i$, and measure their improvement to reach $\tau’$, i.e., $\tau’ - p_i$. Then, we add this tuple to any tuples $(I_{\ell})_{\ell=1}^{g}\in T(\tau',\kappa-1)$, and store all the dominating resulted tuples (the Pareto frontier) in $T(\tau,\kappa)$. 
\begin{algo}
\label{recurrence-exact-fairness-objective}
Run dynamic program based on function $T$, defined below, that takes $\forall \ell \; \cup_{i \in G_\ell} \{p_i\}$ and $k$ as input and outputs $T(\tau_{\min},k)$, as the Pareto-frontier improvement tuples, and $S(\tau_{\min},k)$, as the Pareto-frontier
sets of targets
; where $\tau_{\min}= \min\{\tau \in \T_p\}$ and $\tau_{\max}= \max\{\tau \in \T_p\}$. $T(\tau,\kappa)$ constructs the Pareto frontier for groups' social welfare for agents on or to the right of $\tau\in \T_p$ when at most $\kappa$ target levels can be selected. Function $T$ is defined as follows.
\begin{itemize}
\item[1)]  For any $\tau\in \T_p$, $ \; T(\tau,0)=\vec{0}_g$.

\item[2)]  For any $1\leq \kappa\leq k$, $ \; T(\tau_{\max},\kappa)=\vec{0}_g$.

\item[3)]  For any $\tau\in \T_p, \tau< \tau_{\max}$ and $1\leq \kappa\leq k$:
\[T(\tau,\kappa) = \Bigg\{\Bigg(I_{\ell}+\Big(\sum_{\substack{\tau\leq p_i <\tau'\\\text{ s.t. } \tau'-p_i \leq\Delta_i}} \mathbbm{1}\Big\{i\in G_{\ell}\Big\}(\tau'-p_i)\Big)\Bigg)_{\ell=1}^{g},  \text{ s.t. } (I_{\ell})_{\ell=1}^{g}\in T(\tau',\kappa-1), \tau'\in \mathcal{T}_p, \tau'> \tau \Bigg\}\]
\end{itemize}
$S(\tau,\kappa)$ stores the sets of targets corresponding to the improvement tuples in $T(\tau,\kappa)$. After the above computations, the algorithm removes all the dominated solutions. 
\end{algo}

When all $p_i, \Delta_i$ values are integral, the running time of~\Cref{recurrence-exact-fairness-objective} gets bounded as follows.

\begin{theorem}
\label{prop:running-time-fairness-exact}
\Cref{recurrence-exact-fairness-objective} constructs the Pareto frontier for groups' social welfare using at most $k$ targets given $n$ agents in $g$ groups, and has a running time of $\mathcal{O}(n^{g+2}kg\Delta_{\max}^g)$, where $\Delta_{\max}$ is the maximum improvement capacity.
\end{theorem}
\begin{proof}
See \Cref{app:missing_max_min}.
\end{proof}

\begin{corollary}
\label{cor:exact-fairness}
There is an efficient algorithm that finds a set of at most $k$ targets that maximizes minimum improvement across all groups, i.e., maximizing $\min_{1 \leq \ell \leq g} \sw_\ell$.
\end{corollary}
\begin{proof}
See \Cref{app:missing_max_min}.
\end{proof}

\paragraph{A Fully Polynomial Time Approximation Scheme for the Max-Min Objective.}

The algorithm mentioned in~\Cref{cor:exact-fairness} is pseudo-polytime since its time-complexity depends on the numeric value of $\Delta_{\max}$. We present a Fully Polynomial Time Approximation Scheme (FPTAS) to maximize the minimum improvement across all groups for the setting where each group $G_{\ell}$ has its own improvement capacity $\Delta_{\ell}$. The algorithm finds a set of at most $k$ targets that approximates the max-min objective within a factor of $1-\eps$ for any arbitrary value of $\eps>0$. Here, we relax the assumption that $p_i, \Delta_i$ values need to be integral, and suppose all $p_i, \Delta_i$ values are real numbers.
Similar to the dynamic program based on~\Cref{recurrence-exact-fairness-objective}, for each subproblem, a set containing all $g$-tuples of improvements $(I_1, I_2,\cdots, I_g)$ that are simultaneously achievable for all groups is stored. However, computing all such tuples takes exponential time since $\sum_{i=1}^k \binom{2n}{i}$ possible cases of targets' placements need to be considered. Therefore, we discretize the set of all possible improvements for this problem by rounding all the improvement tuples, and develop an FPTAS algorithm. The recurrence for the dynamic program is given in~\Cref{sec:appendix-FPTAS}. The algorithm runs efficiently when the number of groups is a constant. We defer the technical details to~\Cref{sec:appendix-FPTAS}.

%% file: approx.tex
\section{Simultaneous Approximate Optimality}\label{sec:approx_optimality}
In this section, we establish a structural result about the Pareto optimal solutions, and show there exists a simultaneously approximately optimal solution on the Pareto frontier, where the approximation factor depends on the number of groups. More specifically, given $g$ groups, and limit $k \geq g$ on the number of target levels, we provide \Cref{alg:approx} whose improvement per group is simultaneously an $\Omega(1/g^3)$ approximation of the optimal $k$-target solution for each group; implying a constant approximation when the number of groups is constant. This result is of significance because natural outcomes such as the max-min fair solution and the union of group-optimal targets may lead to arbitrarily poor performance in terms of simultaneous approximate optimality --- See \Cref{ex:max_min_suboptimal,ex:interference}. This result only holds for the common improvement capacity model, and in \Cref{ex:approx_two_groups}, we show such a solution does not exist for the individualized improvement capacity model.

\begin{theorem} \label{thm:approx}
\Cref{alg:approx}, given limit $k \geq g$ on the number of target levels, outputs a solution that is simultaneously $\Omega(1/g^3)$-approximately optimal for each group. More specifically, it provides a solution such that for all $1\leq \ell \leq g$, $\sw_\ell \geq 1/(16g^3) \opt_{\ell}^{k}$, where $\opt_{\ell}^{k}$ is the optimal social welfare of group $\ell$ using at most $k$ target levels.
\end{theorem} 

\begin{corollary} \label{cor:approx}
There is an efficient algorithm to find a simultaneously $\alpha^*$-approximately optimal solution for each group, where $\alpha^*$, defined as the best approximation factor possible, is $\Omega(1/g^3)$. 
\end{corollary}

We are not aware if $\Omega(1/g^3)$ is the best possible ratio, however, the following example shows there are no simultaneously approximately optimal solutions with approximation factor $>1/g$. 

\begin{example}\label{ex:lowerbound}
Let $\Delta = 1$. Suppose group $\ell \in \{1, 2, \ldots, g\}$ has a single agent at position $(\ell-1)/g$; i.e., the agents are at $0, 1/g, \ldots, (g-1)/g$. For each group, the optimal total improvement is $1$ in isolation (independent of the limit on the number of targets). However, using any number of targets in total there are no solutions with $> 1/g$ improvement for all groups. \end{example}

The following example shows that the max-min fair solution does not satisfy a simultaneous constant approximation per group even when there are only two groups.

\begin{example}\label{ex:max_min_suboptimal}
Let $\Delta = 1$. Group $A$ has $n$ agents; one agent at each position $1, 2, \ldots, n$. Group $B$ has $n$ agents in $k$ bundles of size $n/k$. The bundles of agents are at positions $n+1-k^2/n, \ldots, n+k-k^2/n$. The unique max-min solution has targets at $n-k+1, n-k+2, \ldots, n+1$, and leads to $k$ total improvement for each group which is $k/n$ of the optimal total improvement for group $B$.
\end{example}

The following example shows solving the optimization problem separately per group and outputting the union of the targets can lead to arbitrarily low group improvement compared to the optimum.

\begin{example}\label{ex:interference}
Suppose there are two groups $A$ and $B$ and no limit on the number of targets. Group $A$ has $n$ agents at positions $1, 3, 5, \ldots, 2n-1$.  Group $B$ has $n$ agents at positions $2-\eps, 4-\eps, \ldots, 2n-\eps$. {First, consider the common capacity model, where $\Delta = 1$.} In this case, the optimal solution for group $A$ in isolation consists of targets at positions $\{2, 4, \ldots, 2n\}$ and the optimal solution for group $B$ is isolation is $\{3-\eps, 5-\eps, \ldots, 2n+1-\eps\}$. Now, consider a solution that is the union of the targets in the two separate solution. Since each agent in group $B$ is in $\eps$ proximity of a target from group $A$, the total improvement in group $B$ is $n \eps$. Therefore, the total improvement in group $B$ can be arbitrarily close to $0$. {Next, consider the individualized capacity model, where agents in group $A$ have $\Delta_A = 1$, and agents in group $B$ have $\Delta_B = 1+2\eps$. The optimal set of targets in isolation for group $A$ is $\{2, 4, \ldots, 2n\}$, and for group $B$ is $\{3+\eps, 5+\eps, \ldots, 2n+1+\eps\}$. The union of these solutions result in $1+(n-1)\eps$ for group $A$, and $n\eps$ for group $B$ which are arbitrarily low compared to the optimum, which is simultaneously $\geq n(1-\epsilon)$ for group $A$ and $\geq n$ for group $B$.}
\end{example}

The following example shows that if agents can improve by different amounts (the individualized improvement capacity model), then no approximation factor only as a function of $g$ of optimal improvement per group is  possible. 

\begin{example}\label{ex:approx_two_groups}
Suppose groups $A$ and $B$ each have a single agent at position $0$. The agent in group $A$ has improvement capacity $\Delta_A = \eps$ and the agent in group $B$ has improvement capacity $\Delta_B = 1$. The optimal total improvement in isolation for group $A$ is $\epsilon$, and for group $B$ is $1$. However, when considering both groups, no placement of targets with positive improvement for group $A$ leads to $> \eps$ improvement for group $B$.
\end{example}

First, we describe a high-level overview of \Cref{alg:approx}. The algorithm proceeds in the following four main steps. 

\begin{enumerate}
    \item \textbf{Optimal targets in isolation.} Run~\Cref{def:recurrence-dp-one-group} separately for each group to find an optimal allocation of at most $\lceil k/g \rceil$ targets \footnote{Although the total number of targets used in this step can be more than $k$, after the algorithm ends at most $k$ targets are being used in total.}. Let $\T_\ell$ be the output for group $\ell$. 
    \item \textbf{Distant targets in isolation.} Delete $3/4$ fraction of each set of target levels, $\T_\ell$, such that (1) the distance between every two consecutive targets in each set is at least $2\Delta$ and (2) the new $\T_\ell$  (after deletion) guarantees an $\Omega(1)$ approximation of the previous step when the targets for each group are considered in isolation. {\Cref{sec:step2} below shows this is possible.}
    \item \textbf{Locally optimized distant targets in isolation.} For each $\ell$ and $\tau \in \T_\ell$, consider the agents in group $\ell$ that afford to reach $\tau$ (agents in $G_{\ell}$ $\cap[\tau-\Delta, \tau)$). Optimize $\tau$ to maximize the total improvement for this set of agents.
    \item \textbf{Resolve interference of targets.} Consider sets of interfering targets. Relocate these targets locally to guarantee $\Omega(1/g^2)$ approximation per group compared to the previous step where each group was considered in isolation. {\Cref{sec:step4} below shows this is possible.}
\end{enumerate}

\begin{algorithm}[!ht]
    \SetNoFillComment
    \SetAlgoLined
    \DontPrintSemicolon
    \For{ $\ell= 1 \text{ to } g$}{
        \tcc{Step 1}
        Let $\T_\ell: \tau_1 <\tau_2< \ldots$  be the output of~\Cref{def:recurrence-dp-one-group} for agents in $G_{\ell}$ and limit $\lceil k/g \rceil$ on the number of targets.\;

        \tcc{Step 2}
        Partition $\T_\ell$ to $4$ parts $P_1, P_2, P_3, P_4$, where $P_i:= \tau_i, \tau_{4+i}, \tau_{8+i}, \ldots$.\; 
        Update $\T_\ell$ by keeping the part with the highest improvement and deleting the rest.\;
        
        \tcc{Step 3}
        Delete agents in $G_\ell$ that do not improve given $\T_\ell$.\;
        For all $\tau \in \T_\ell$, replace $\tau$ with the output of~\Cref{def:recurrence-dp-one-group} for agents in $[\tau-\Delta, \tau) \cap G_{\ell}$ and limit $1$ on the number of targets.\; 
    }
    \tcc{Step 4}
    $\T: \tau_1 < \tau_2 < \ldots = \cup_\ell \T_\ell$\\
    $S, \T^* = \emptyset$\\
    \For{$\tau_j \in \T$}{
        $s_j = \tau_j - \Delta$\\
        $S = S \cup \{s_j\}$\\
    }
    {Partition $S : s_1 < s_2 < \ldots$ into the least number of parts of consecutive points: $S_1, S_2, \ldots$, such that in each part, $S_i$, each two consecutive points are at distance less than $\Delta/g$.}\;
    \For{all $S_i : s_u < s_{u+1} < \ldots < s_v$}{
        $\tau^*_i = \min\{\tau_u, s_{v+1}\}$.\\
        $\T^* = \T^* \cup \tau^*_i$.\\
    }
    \Return $\T^*$\\
    \caption{Simultaneous approximate optimality per group.}
    \label{alg:approx}
\end{algorithm}

Now, we describe and analyze these steps in more detail. 

\subsection{Step $1$: Optimal targets in isolation}

At the end of step $1$, $\T_\ell$ is the optimal set of targets for $G_\ell$ in isolation. The following observation shows that without loss of optimality, we may assume the distance between every other target level is at least $\Delta$.\footnote{\Cref{ex:consecutive_less_than_delta}, however, shows the distance between two \emph{consecutive} targets may be arbitrarily smaller than $\Delta$.}

\begin{observation}\label{obs:delta_apart}
Consider a set of target levels $\T: \tau_1 < \tau_2 < \ldots$.  Suppose $\tau_{j+2} < \tau_j + \Delta$. By removing $\tau_{j+1}$, any agent with $\tau_j \leq p_i < \tau_{j+1}$ improves strictly more, and other agents improve the same amount. This weakly increases social welfare. 
\end{observation} 

\subsection{Step $2$: Distant targets in isolation}\label{sec:step2}

Step $2$ of the algorithm runs the following procedure for $\T_\ell$.

\begin{definition}[Distant targets procedure] \label{def:distant_procedure}
Consider solution $\T: \tau_1< \tau_2< \ldots$, where for all $j$,  $\tau_{j+2}-\tau_j \geq \Delta$ as input to the following procedure.
\begin{itemize}
    \item Partition $\T$ into $4$ parts, $P_1, P_2, P_4, P_4$, where $P_i=: \tau_i, \tau_{4+i}, \tau_{8+i}, \ldots$. Consider the part $P_i$ that introduces the highest improvement. Update $\T$ to $P_i$ (and delete the rest).
\end{itemize}
\end{definition}

The following lemma shows that at the end of this step, target levels in $\T_\ell$ are $2\Delta$ apart, this step provides a $4$-approximation compared to the previous step, and the number of targets designated to each group is at most $\lfloor k/g \rfloor$.

\begin{lemma}\label{lm:make_distant}
Consider solution $\T: \tau_1< \tau_2< \ldots$ with total improvement $I$ such that for all $j$,  $\tau_{j+2}-\tau_j \geq \Delta$. Consider the procedure in \Cref{def:distant_procedure}. This procedure results in a solution $\T': \tau'_1< \tau'_2< \ldots$ where $\forall j \; \tau'_{j+1}-\tau'_j \geq 2\Delta$, has total improvement at least $ I/4$, and $|\T'| \leq \lceil |\T|/4 \rceil$. Particularly, for $|\T|\leq \lceil k/g \rceil$ where $k \geq g$, the number of final targets, $|\T'|$, is at most $\lfloor k/g \rfloor$.
\end{lemma}
\begin{proof} See \Cref{app:approximation}.
\end{proof}

\subsection{Step $3$: Locally optimized distant targets in isolation}

At the end of step $2$, every two targets in $\T_\ell$, the set of targets for group $\ell$, are at distance at least $2\Delta$. Consider only the targets and agents in group $\ell$. For each $\tau \in \T_\ell$, agents in $[\tau-\Delta, \tau)$ improve to $\tau$ and the remaining agents do not improve. 
To continue with the algorithm, we first delete the agents that do not improve. Then, we optimize $\T_\ell$ for the set of agents that do improve. This modification is necessary for the next step. To do the optimization, we use~\Cref{def:recurrence-dp-one-group} for agents in $[\tau-\Delta, \tau)$ for any $\tau \in \T_\ell$ and limit $1$ on the number of targets, and replace $\tau$ with the output of the algorithm.

\begin{lemma}\label{lm:step_three}
At the end of step $3$ in \Cref{alg:approx}, (i) the distance between every two targets in $\T_\ell$ is at least $\Delta$; (ii) each target $\tau \in \T_\ell$ is optimal, i.e., maximizes total improvement for agents in $G_\ell \cap [\tau-\Delta, \tau)$; and (iii) the total amount of improvement of $G_\ell$ using solution $\T_\ell$ does not decrease compared to the previous step.
\end{lemma}
\begin{proof} See \Cref{app:approximation}.
\end{proof}

Now, we extract properties about optimal solutions. Since at the end of step $3$, $\T_\ell$ is optimal for $G_\ell$ we take advantage of these properties in the remaining steps of the algorithm. 

The following lemma shows that if $\tau$ is optimal for agents in $[\tau-\Delta, \tau)$, a considerable fraction of these agents reside in the left-most part of the interval.

\begin{lemma}\label{lm:p_x}
Consider optimal target $\tau$ for the set of agents $A$ in $[\tau-\Delta, \tau)$ in absence of other targets. For each $0\leq x \leq 1$, at least $x$ fraction of $A$ belong to $[\tau-\Delta,\tau-\Delta + x\Delta)$. In particular, at least $1/(2g)$ fraction of the agents are in $[\tau-\Delta,\tau-(2g-1)/(2g) \Delta)$. 
\end{lemma}
\begin{proof}
Let $p_x$ be the fraction of agents in $A$ in $[\tau-\Delta, \tau-\Delta+x\Delta)$. Each of these agents is improving by at least $(1-x) \Delta$. Therefore, the contribution of these agents to total improvement of $A$ is at least $p_x|A| (1-x)\Delta$. Since $\tau$ is the optimal target, it introduces at least as much improvement as any other target, and in particular a target at $\tau' = \tau+x$. Consider the total improvement introduced by $\tau'$ compared to $\tau$ (in absence of target $\tau$). The contribution of the agents in $[\tau-\Delta, \tau-\Delta+x\Delta)$ to total improvement reduces to $0$, but the contribution of the agents in $[\tau-\Delta+x\Delta, \tau)$ increases by $(1-p_x)|A|x\Delta$. Since $\tau$ is the optimal target, the loss of substituting it with $\tau'$ is at least as much as the gain. Therefore, $p_x (1-x)\Delta \geq (1-p_x)x\Delta$; which implies $p_x \geq x$. 
\end{proof}

The following lemma shows that if $\tau$ is optimal for agents in $[\tau-\Delta, \tau)$, substituting $\tau$ with another target in this interval, far enough from the left endpoint, $\tau-\Delta$, guarantees a considerable fraction of the optimal improvement.

\begin{lemma}\label{lm:shift}
Consider optimal target $\tau$ for agents $A$ in $[\tau-\Delta, \tau)$ in absence of other targets. By relocating $\tau$ to any point in $[\tau-\Delta+x\Delta, \tau]$, for $0 \leq x \leq 1$, the total improvement of $A$ is at least $x^2/4$ of the optimum. In particular, by relocating $\tau$ to any point in $[\tau-\Delta+\Delta/g, \tau]$, the total improvement is at least $1/(4g^2)$ of the optimum.
\end{lemma}
\begin{proof}
Similar to the previous lemma, let $p_{x/2}$ be the fraction of agents in $[\tau-\Delta,\tau-\Delta+(x/2)\Delta)$. After the relocation, each such agent improves by at least $(x/2)\Delta$; therefore, the contribution of these agents to total improvement is at least $p_{x/2}|A|(x/2)\Delta$. The optimal total improvement is bounded by $|A|\Delta$. Therefore, using $p_{x/2} \geq x/2$, by \Cref{lm:p_x}, the total improvement after relocation is at least $x^2/4$ of the optimum.
\end{proof}

\subsection{Step $4$: Resolve interference of targets}\label{sec:step4}

In this step, we consider the solutions for all groups together and resolve the interference of targets designed for different groups. As illustrated in \Cref{ex:interference}, this interference can lead to arbitrarily low social welfare. To resolve this issue, we take advantage of sparsity of the targets designed for the same group (step $2$) and optimality of $\T_\ell$ for $G_\ell$ (step $3$). 

The main purpose of this step is to recover an approximation guarantee of the total improvement of \emph{each target in isolation} at the end of step $3$ by removing the interference among the targets. Particularly, for each target $\tau \in \T_\ell$ in isolation, we consider agents in $G_\ell$ reaching to that, i.e., agents in interval $[\tau-\Delta, \tau)$. By \Cref{lm:p_x}, a considerable fraction of these agents are on the left-most side of the interval. And as shown in \Cref{lm:shift}, as long as there exists a target far enough from the left endpoint we are in good shape. More precisely, if for all $\tau$ at the beginning of this step, there is a target in the final solution in $[\tau-\Delta+\Delta/g, \tau]$ (property $1$), and no targets in $(\tau-\Delta, \tau-\Delta+\Delta/g)$ (property $2$), a $1/(4g^2)$ fraction is achievable. The set of targets at the end of step $3$ may fail to satisfy these properties, because there may be targets $\tau' < \tau$ such that $\tau'$ is not far enough from the left endpoint of the interval corresponding to $\tau$; i.e., for $s=\tau-\Delta, \; s < \tau' < s + \Delta/g$. 

To resolve the interference among the targets, in step $4$, we work as follows. First, we consider the left endpoints of improvement intervals corresponding to the targets; i.e., $\forall \tau_j$, at the end of step $3$, consider $s_j=\tau_j-\Delta$. Then, we partition these left endpoints into maximal parts $S_1, S_2, \ldots$, such that in each part, the distance between every two consecutive points is small, particularly, less than $\Delta/g$. Using the sparsity of the targets (step $2$) the number of points in each part is bounded. Finally, we design a new target $\tau^*_i$ (defined formally below) corresponding to part $S_i$, such that $\tau^*_i$ is to the left of any $S_j$ with $j > i$, and at distance between $\Delta/g$ and $\Delta$ to the right of the points in $S_i$ (satisfying properties $1$ and $2$). Using optimality of $\T_\ell$ for $G_\ell$ (step $3$) this results in the desired approximation factor.

More formally, this step proceeds as follows.

\begin{enumerate}
    \item Let $\T: \tau_1 < \tau_2 < \ldots$ be the union of the set of targets found at the end of step $3$.
    \item Construct $S: s_1 < s_2 < \ldots$ from $\T$, such that $\forall \tau_j \in \T$, include $s_j = \tau_j - \Delta$ in $S$.
    \item Partition $S$ into the least number of parts of consecutive points: $S_1, S_2, \ldots$, such that in each part $S_i : s_u < s_{u+1} < \ldots < s_v$, each two consecutive points are at distance less than $\Delta/g$; i.e., $\forall s_r,s_{r+1} \in S_i, s_{r+1}-s_r < \Delta/g $. By construction of the first three steps (and as shown in the proof of \Cref{lm:step_four}), the number of points in each part is at most $g$.
    \item For each $S_i : s_u < s_{u+1} < \ldots < s_v$, consider new target $\tau^*_i = \min\{\tau_u, s_{v+1}\}$. 
    \item Output the set of new targets.
\end{enumerate}

\begin{lemma}\label{lm:step_four}
Consider $\T$ as the union of all solutions at the end of step $3$. For all $\tau \in \T$, consider the interval $[\tau-\Delta, \tau)$ which consists of agents that improve to target $\tau$ if it were the only target available. At the end of step $4$, (i) there will be a target in $[\tau-\Delta+\Delta/g, \tau]$, and (ii) there will be no targets in $(\tau-\Delta, \tau-\Delta+\Delta/g)$.
\end{lemma}

\begin{proof} See \Cref{app:approximation}.
\end{proof}

\subsection{Putting everything together}

\begin{theorem} \label{thm:modified}
\Cref{alg:approx}, given $k \geq g$, provides a solution with at most $k$ number of targets, such that for all $1\leq \ell \leq g$, $\sw_\ell \geq 1/(16g^2) \opt_{\ell}^{\lceil k/g \rceil}$, where $\opt_{\ell}^{k}$ is the optimal social welfare of group $\ell$ using at most $k$ target levels.
\end{theorem}

\begin{proof}
By \Cref{obs:delta_apart} and \Cref{lm:make_distant}, when the targets designed for each group are considered separately and in isolation, at the end of step $2$, there are at most $\lfloor k/g \rfloor$ targets designed for group $\ell$ and the total improvement in this group is $1/4$-approximation of $\opt_{\ell}^{\lceil k/g \rceil}$. By \Cref{lm:step_three}, \Cref{lm:shift}, and \Cref{lm:step_four}, we lose another $4g^2$ factor compared to step $2$. In total, \Cref{alg:approx} results in $\sw_\ell \geq 1/(16g^2) \opt_{\ell}^{\lceil k/g \rceil}$, for all groups $1\leq \ell \leq g$. Also, when $k \geq g$, the total number of targets is at most $g \lfloor k/g \rfloor \leq k$.
\end{proof}

\begin{proof}[Proof of \Cref{thm:approx}]
Given \Cref{thm:modified}, it suffices to argue $\opt_{\ell}^{\lceil k/g \rceil} \geq  \opt_{\ell}^{k}/g$; i.e., when the number of targets increases by a factor, here $g$, the optimal total improvement increases by at most that factor. This statement is straightforward using subadditivity of total improvement as a function of the set of targets. Specifically, consider the optimal $k$-target solution and an arbitrary partition with $g$ parts of size $\lceil k/g \rceil$ or $\lfloor k/g \rfloor$; by subadditivity, one of the parts provides at least $1/g$ of the total improvement. 
\end{proof}

\begin{proof}[Proof of \Cref{cor:approx}] \Cref{recurrence-exact-fairness-objective} in \Cref{sec:max_min} outputs the Pareto frontier for groups' social welfare. By definition, the solution provided in \Cref{alg:approx} is dominated by a solution on the Pareto frontier. By computing the factor of simultaneous approximate optimality of each solution on the Pareto frontier, we find the solution that achieves the best simultaneous approximation factor $\alpha^3$, and by \Cref{thm:approx}, this solution is simultaneously $\Omega(1/g^3)$-approximately optimal.
\end{proof}

\begin{remark}[a weaker benchmark and a tighter gap]
In contrast with \Cref{thm:approx} that measures the performance of \Cref{alg:approx} with respect to the optimal $k$-target solution for each group (the notion of simultaneous approximate optimality), \Cref{thm:modified} measures the performance with respect to the optimal $\lceil k/g \rceil$-target solution for each group. Since the lower bound provided in \Cref{ex:lowerbound} shows achieving better than $1/g$ of either of these benchmarks is not possible, there is only a factor $g$ gap in the performance of the algorithm and the lower bound with respect to the optimal $\lceil k/g \rceil$-target solution. 
\end{remark}

%% file: learning.tex
\section{Generalization Guarantees}\label{sec:gen_guarantees}
In this section, we generalize our results to a setting where we only have sample access to agents and provide sample complexity results. \Cref{sec:gen_single} provides a guarantee for the maximization objective in absence of fairness, and \Cref{sec:gen_single} provides a guarantee for the fairness objectives.

\subsection{Generalization Guarantees For the Maximization Objective}
\label{sec:gen_single}
Suppose there is a distribution $\mathcal{D}$ over agents' positions. Our goal is to find a set of $k$ targets $\T$ that maximizes expected improvement of an agent when we only have access to $n$ agents sampled from $\mathcal{D}$. For any distribution $\mathcal{D}$ over agents' positions, we define $I_{\mathcal{D}}(\mathcal{T}) = \E_{p\sim \mathcal{D}}[I_p(\mathcal{T})]$, where $I_p(\mathcal{T})$ captures the improvement of agent $p$ given the targets in $\mathcal{T}$. 
In~\Cref{thm:generalization-absent-fairness}, we provide a generalization guarantee that shows if we sample a set $S$ of size $n\geq \eps^{-2}\big(\Delta_{\max}^2(k\ln(k)+\ln(1/\delta))\big)$ drawn \emph{i.i.d} from $\mathcal{D}$, then with probability at least $1-\delta$, for all sets $\T$ of $k$ targets, we can bound the difference between average performance over $S$ and actual expected performance, such that $\big|I_S(\mathcal{T})-I_{\mathcal{D}}(\mathcal{T})\big|\leq \mathcal{O}(\eps)$. Formally, we show the following theorem holds:
\begin{theorem}
(Generalization of the maximization objective) Let $\mathcal{D}$ be a distribution over agents' positions. For any $\eps>0$, $\delta>0$, and number of targets $k$, if $S=\{p_i\}_{i=1}^n$ is drawn i.i.d.\ from $\mathcal{D}$ where 
$n\geq \eps^{-2}\Delta_{\max}^2\big(k\ln(k)+\ln(1/\delta)\big)$,
then with probability at least $1-\delta$, for all sets $\T$ of $k$ targets, $\big|I_S(\mathcal{T})-I_{\mathcal{D}}(\mathcal{T})\big|\leq \mathcal{O}(\eps)$.
\label{thm:generalization-absent-fairness}
\end{theorem}
In particular, the solution $\T^*$ that maximizes improvement on $\mathcal{S}$, also maximizes improvement on $\mathcal{D}$ within an additive factor of $\mathcal{O}(\eps)$.

In order to prove ~\Cref{thm:generalization-absent-fairness}, we use two main ideas. First, using a framework developed by Balcan et al.~\cite{Balcan-STOC21}, we bound the \emph{pseudo-dimension} complexity of our improvement function. Then, using classic results from learning theory~\cite{pollard1984convergence}, we show how to
translate \emph{pseudo-dimension} bounds into generalization guarantees. The framework proposed by Balcan et al.~\cite{Balcan-STOC21} depends on the relationship between primal and dual functions. When the dual function is piece-wise constant, piece-wise linear or generally piece-wise structured, they show a general theorem that bounds the \emph{pseudo-dimension} of the primal function.
Formally \emph{pseudo-dimension} is defined as following:

\begin{definition}
(Pollard's Pseudo-Dimension) A class $\mathcal{F}$ of real-valued functions $P$-shatters a set of points $\mathcal{X} = \{x_1, x_2,\cdots, x_n\}$ if there exists a set of thresholds $\gamma_1, \gamma_2,\cdots, \gamma_n$ such that for every subset $T\subseteq \mathcal{X}$, there exists a function $f_T\in \mathcal{F}$ such that $f_T(x_i)\geq \gamma_i$ if and only if $x_i\in T$. In other words, all $2^n$ possible above/below patterns are achievable for targets $\gamma_1,\cdots, \gamma_n$. The pseudo-dimension of $\mathcal{F}$, denoted by $\PDim(\mathcal{F})$, is the size of the largest set of points that it $P$-shatters.
\end{definition}

Balcan et al.~\cite{Balcan-STOC21} show when the dual function is piece-wise structured, the \emph{pseudo-dimension} of the primal function gets bounded as following:
\begin{theorem}
(Bounding Pseudo-Dimension~\cite{Balcan-STOC21})
Let $\mathcal{U}=\{u_{\vec{\rho}}\mid \vec{\rho}\in \mathcal{P}\subseteq \mathbb{R}^d\}$ be a class of utility functions defined over a $d$-dimensional parameter space. Suppose the dual class $\mathcal{U}^*$ is $(\mathcal{F},\mathcal{G},m)$-piecewise decomposable, where the boundary functions $\mathcal{G}=\{f_{\vec{a}, \theta}: \mathcal{U}\rightarrow \{0,1\}\mid \vec{a}\in \mathbb{R}^d, \theta\in \mathbb{R}\}$ are halfspace indicator functions $g_{\vec{a},\theta}:u_{\rho}\rightarrow \mathbb{I}_{\vec{a}\cdot\vec{\rho}\leq \theta}$ and the piece functions $\mathcal{F}=\{f_{\vec{a},\theta}:\mathcal{U}\rightarrow \mathbb{R}\mid \vec{a}\in \mathbb{R}^d, \theta\in \mathbb{R}\}$ are linear functions $f_{\vec{a},\theta}:u_{\rho}\rightarrow \vec{a}\cdot \vec{\rho}+\theta$, and $m$ shows the number of boundary functions. Then, $\PDim(\mathcal{U}) = \mathcal{O}(d\ln(dm))$.
\label{thm-Balcan-Pdim}
\end{theorem}

We use~\Cref{thm-Balcan-Pdim} to bound the \emph{pseudo-dimension} of the improvement function.

\begin{lemma}
\label{lem:bounding-pdim}
Let $\mathcal{U}=\{u_{\mathcal{T}}:p\rightarrow u_{\mathcal{T}}(p) \mid \mathcal{T}\in \mathbb{R}^k, p\in \mathbb{R}\}$ be a set of functions, where each function 
defined by a set of $k$ targets, takes as input a point $p\in \mathbb{R}$ that captures an agent's position, and outputs a number showing the improvement that the agent can make. Then, $\PDim(\mathcal{U})=\mathcal{O}(k\ln(k))$.
\end{lemma}

\begin{proof}
We use Theorem~\ref{thm-Balcan-Pdim} to bound $\PDim(\mathcal{U})$. First, we define the dual class of $\mathcal{U}$ denoted by $\mathcal{U}^*$. The function class $\mathcal{U}^*=\{u^*_p:\mathcal{T}\rightarrow u_{p}(\mathcal{T}) \mid \mathcal{T}\in \mathbb{R}^k, p\in \mathbb{R}\}$ is a set of 
functions, where each function defined 
by an agent $p$, takes as input a set $\mathcal{T}\in \mathbb{R}^k$ of $k$ targets \footnote{If the input consists of $k'$ targets where $k'<k$, it resembles the case where $k$ targets are used and $k-k'$ of them are ineffective, i.e., are put at position $\tau_{\min}$.}, and outputs the improvement that $p$ can make given $\mathcal{T}$. 
Geometrically, in the dual space, there are $k$ dimensions $\tau_1,\cdots,\tau_k$, and each dimension is corresponding to one target. In order to use~\Cref{thm-Balcan-Pdim}, we show that $\mathcal{U}^* = (\mathcal{F},\mathcal{G},k)$ is piecewise-structured. The boundary functions in $\mathcal{G}$ are defined as follows. 
If agent $p$ improves to a target $\tau_i$, then $0< \tau_i - p \leq \Delta$, where $\Delta$ is the improvement capacity of $p$. Additionally, between all the targets within a distance of at most $\Delta$, $p$ improves to the closest one. For each pair of integers $(i,j)$, where $1\leq i,j\leq k$, we add the hyperplane $\tau_i -\tau_j =0$ to $\mathcal{G}$. Above this hyperplane is the region where $\tau_i >\tau_j$, implying that $\tau_i$ comes after $\tau_j$. Below the hyperplane is the region where the ordering is reversed. In addition, for each target $\tau_i$, we add the boundary functions $\tau_i = p$ and $\tau_i = p+\Delta$ to $\mathcal{G}$. In the region between $\tau_i = p$ and $\tau_i = p+\Delta$, $\tau_i$ is effective and the agent can improve to it.
Now, the dual space is partitioned into a set of regions. In each region, either there exists a unique closest effective target $(\tau_r)$, or all the targets are ineffective. In the former case, the improvement that the agent makes is a linear function of its distance from the closest effective target $(f = \tau_r-p)$. In the later case, the agent makes no improvement ($f = 0$). Therefore, the piece functions in $\mathcal{F}$ are either constant or linear. Now, since the total number of boundary functions is $m=\mathcal{O}(k^2)$ and the space is $k$-dimensional, using~\Cref{thm-Balcan-Pdim},  $\PDim(\mathcal{U})$ is $\mathcal{O}(k\ln(k^3)) = \mathcal{O}(k\ln(k))$.
\end{proof}

Now, we are ready to prove~\Cref{thm:generalization-absent-fairness}.

\begin{proof}[Proof of~\Cref{thm:generalization-absent-fairness}]
Classic results from learning theory~\cite{pollard1984convergence} 
show the following generalization guarantees: Suppose $[0,H]$ is the range of functions in hypothesis class $\mathcal{H}$. For any $\delta\in(0,1)$, and any distribution $\mathcal{D}$ over $\mathcal{X}$, with probability $1-\delta$ over the draw of $\mathcal{S}\sim \mathcal{D}^n$, for all functions $h\in \mathcal{H}$, the difference between the average value of $h$ over $\mathcal{S}$ and its expected value gets bounded as follows:
\[\Big|\frac{1}{n}\sum_{x\in \mathcal{S}}h(x)-\E_{y\sim \mathcal{D}}[h(y)]\Big| = \mathcal{O}\Big(H\sqrt{\frac{1}{n}\Big(\PDim(\mathcal{H})+\ln(\frac{1}{\delta})\Big)}\Big)\]

In the case of maximizing improvement, $H=\Delta_{max}$ and $\PDim(\mathcal{H}) = \mathcal{O}(k\ln(k))$. By setting $n\geq \eps^{-2}\Delta_{max}^2\big(k\ln(k)+\ln(1/\delta)\big)$, with probability at least $1-\delta$, the difference between the average performance over $\mathcal{S}$ and the expected performance on $\mathcal{D}$ gets upper-bounded by $\mathcal{O}(\eps)$. 
\end{proof}

\subsection{Generalization Guarantees For Fairness Objectives}\label{sec:gen_multiple}

Suppose there is a distribution $\mathcal{D}_{\ell}$ of agents' positions for each group $\ell$. Let $\mathcal{D}=\sum_{\ell=1}^g \alpha_{\ell}\mathcal{D}_{\ell}$ be a weighted mixture of distributions $\mathcal{D}_1,\cdots,\mathcal{D}_g$. Let $\alpha_{\min} = \min_{1\leq \ell \leq g}\alpha_{\ell}$. Suppose we have sampling access to $\mathcal{D}$ and cannot directly sample from $\mathcal{D}_1,\cdots,\mathcal{D}_g$. Our goal is to derive generalization guarantees for different objective functions across multiple groups 
when we only have access to a set $S$ of $n$ agents sampled from distribution $\mathcal{D}$.
Let $I_{G_{\ell}}(\mathcal{T})$ denote the average improvement of agents in group $G_{\ell}\subseteq S$ given a set $\mathcal{T}$ of $k$ targets. Let $I_{\mathcal{D}_{\ell}}(\mathcal{T}) = \E_{p\sim \mathcal{D}_{\ell}}[I_p(\mathcal{T})]$, where $I_p(\mathcal{T})$ captures the improvement of agent $p$ given $\mathcal{T}$.
In~\Cref{thm:generalization-fairness}, we show if we sample a set $S$ of $\mathcal{O}\Big(\alpha_{\min}^{-1}\Big(\eps^{-2}\Delta_{\max}^2\big(k\ln(k)+\ln(g/\delta)\big)+\ln(g/\delta)\Big)\Big)$ examples drawn \emph{i.i.d.}\ from $\mathcal{D}$, then for all sets $\mathcal{T}$ of $k$ targets and for all groups $\ell$, $\big|I_{G_{\ell}}(\mathcal{T})-I_{\mathcal{D}_{\ell}}(\mathcal{T})\big|\leq \mathcal{O}(\eps)$.

\begin{theorem}
(Generalization across multiple groups) Let $\mathcal{D}$ be a distribution over agents' positions. For any $\eps>0$, $\delta>0$, and number of targets $k$, if $S=\{p_i\}_{i=1}^n$ consisting of $g$ groups $\{G_{\ell}\}_{\ell=1}^g$ is drawn i.i.d.\ from $\mathcal{D}$,
where $n\geq (2/\alpha_{\min})\big(\eps^{-2}\Delta_{\max}^2(k\ln(k)+\ln(2g/\delta))+4\ln(2g/\delta)\big)$,
then with probability at least $1-\delta$, for all sets $\mathcal T$ of $k$ targets, for all groups $\ell$, $\big|I_{G_{\ell}}(\mathcal{T})-I_{\mathcal{D}_{\ell}}(\mathcal{T})\big|\leq \mathcal{O}(\eps)$.
\label{thm:generalization-fairness}
\end{theorem}

\begin{proof}
Let $S$ be partitioned into $g$ groups where each group $G_{\ell}$ has size $n_{\ell}$. 
First, for each group $\ell$, let $A_{\ell}$ denote the event that $n_{\ell}\geq (n\alpha_{\ell})/2$. Using Chernoff-Hoeffding bounds we have $\Pr[n_{\ell}<(n\alpha_{\ell})/2]\leq e^{(-n\alpha_{\ell})/8}\leq \delta/(2g)$. The last inequality holds since $n\geq 8\ln(2g/\delta)/\alpha_{\ell}$. Next, for each group $\ell$, let $B_{\ell}$ denote the event that  $\big|I_{G_{\ell}}(\mathcal{T})-I_{\mathcal{D}_{\ell}}(\mathcal{T})\big|\leq \mathcal{O}(\eps)$, then: 
\begin{align}
\Pr[B_{\ell}]\geq \Pr[B_{\ell}\cap A_{\ell}] = \Pr[B_{\ell}\mid A_{\ell}]\cdot \Pr[A_{\ell}]\geq (1-\delta/(2g))(1-\delta/(2g))\geq (1-\delta/g)
\label{bound-prob-B-ell}
\end{align}
In the above statement, inequality $\Pr[B_{\ell}\mid A_{\ell}]\geq (1-\delta/(2g))$ holds since given $A_{\ell}$ happens, then $n_{\ell}\geq \eps^{-2}\Delta_{max}^2(k\ln(k)+\ln(2g/\delta))$, and by~\Cref{thm:generalization-absent-fairness}, event $B_{\ell}$ happens with probability at least $1-\delta/(2g)$. Now, by~\Cref{bound-prob-B-ell}, $\Pr[B_{\ell}]\geq 1-\delta/g$. By applying a union bound, event $B_{\ell}$ happens with probability at least $1-\delta$ for any group $\ell$. 
\end{proof}

In particular, solution $\T^*$ satisfying one of the fairness notions considered in this paper, e.g., simultaneous approximate optimality or maxmizing minimum improvement across groups, on input $S$, achieves a performance guarantee within an additive factor of $\mathcal{O}(\eps)$ on inputs drawn from $\mathcal{D}$.

%% file: extensions.tex
\section{Extensions and Open Problems}\label{sec:extensions}

This section provides two extensions to our objective function: 1) maximizing social welfare subject to a lower bound on the number of improving agents, and 2) optimizing the number of target levels. The section concludes with our main open problem of optimizing the factor of simultaneous approximate optimality and tightening the gap between the upper and lower bounds.

\subsection{Extension 1: A lower bound on the number of agents that improve}

Consider~\Cref{def:recurrence-dp-one-group} whose goal is to find a set of at most $k$ target levels that maximizes the total improvement for a collection of $n$ agents. It is possible that the solution of this algorithm focuses on a small fraction of the agents and does not help many agents to improve. In~\Cref{def:recurrence-dp-extension}, we show how to modify~\Cref{def:recurrence-dp-one-group} to ensure at least $n_{\ell b}$ agents improve. 
The main idea for the recursive step (item $4$ in \Cref{def:recurrence-dp-extension}) is to first consider the potential leftmost targets $\tau' > \tau$, let $x$ denote the number of agents that are within reach to $\tau'$,  and use the smaller subproblem of finding the optimal targets for agents on or to the right of $\tau'$ with one less available target level and an updated lower bound of $\eta-x$, i.e., $S(\tau', \kappa-1, \eta-x)$. We add the performance of each potential leftmost target to the optimal improvement of the remaining subproblem and pick the leftmost target that maximizes this summation.

\begin{algo}
\label{def:recurrence-dp-extension}
Run dynamic program based on function $S$, defined below, that takes $\cup_i \{p_i\}$ and $k$ as input and outputs $S(\tau_{\min},k,n_{\ell b})$, as the optimal improvement, and $S'(\tau_{\min},k,n_{\ell b})$, as the optimal set of targets; where $\tau_{\min}= \min\{\tau \in \T_p\}$ and $\tau_{\max}= \max\{\tau \in \T_p\}$. $S(\tau,\kappa,\eta)$ captures the maximum improvement possible for agents on or to the right of $\tau\in \T_p$ when $\kappa$ target levels can be selected and at least $\eta$ agents need to improve. If $S(\tau_{\min},k,n_{\ell b})=-\infty$ then incentivizing at least $n_{lb}$ agents to improve is impossible. Function $S$ is defined as follows.
\begin{itemize}

\item[1)]  For any $\tau\in \T_p, \eta\geq 1$, we have $S(\tau,0,\eta)=-\infty$.

\item[2)]  For any $1\leq \kappa\leq k, \eta\geq 1$, $S(\tau_{max},\kappa,\eta)=-\infty$, where $\tau_{max}=\max\{\tau\in \T_p\}$. This holds since no agents can improve to $\tau_{max}$, however at least $\eta$ agents to the right of $\tau_{max}$ need to improve which is a contradiction.

\item[3)]  For any $\tau\in \T_p, 0\leq \kappa\leq k, \eta\leq 0$, $S(\tau,\kappa,\eta)=T(\tau,\kappa)$ where function $T$ is defined in~\Cref{def:recurrence-dp-one-group}.

\item[4)]  For any $\tau\in \T_p, \tau< \tau_{max}$, $1\leq \kappa\leq k$, and $1\leq \eta\leq n$:

\begin{gather*} 
S(\tau,\kappa,\eta) = \max_{\tau'\in \Tau_p \ \text{s.t} \ \tau' > \tau}\Bigg(S(\tau',\kappa-1,\eta-\mathbbm{1}\big[i\mid \tau\leq p_i <\tau' \text{ s.t. } \tau'-p_i \leq\Delta_i\big]) \ + \sum_{\tau\leq p_i <\tau'\text{ s.t. } \tau'-p_i \leq\Delta_i}(\tau'-p_i)\Bigg)
\end{gather*} 

\end{itemize}
$S'(\tau,\kappa,\eta)$ keeps track of the optimal set of targets corresponding to $S(\tau,\kappa,\eta)$.
\end{algo}

\subsection{Extension 2: Optimizing the number of target levels}

The nonmonotonicity property may make adding a new target level to the current placement reduce the maximum improvement (see~\Cref{fig:nonmonotone_sub}), or wasteful if we place the new target level somewhere no agent can reach or on top of an existing target. Therefore, when considering $k = 1, 2, \ldots, n$, it is possible that the maximum total improvement is achieved at $k < n$. Using the dynamic program based on~\Cref{def:recurrence-dp-one-group} we can find the minimum value of $k$ that satisfies this property and minimizes the number of targets subject to achieving maximum total improvement. Furthermore, by finding the total amount of improvement for different values of $k$, the principal can decide how many targets are sufficient to achieve a desirable total improvement (bi-criteria objective).

\subsection{Open Problem: Tightening the approximation gap}

\Cref{alg:approx}, as stated in \Cref{thm:approx}, provides an $\Omega(1/g^3)$-approximation simultaneous guarantee compared to the optimal solution for each group using at most $k$ targets; and as stated in \Cref{thm:modified}, provides an $\Omega(1/g^2)$-approximation simultaneous guarantee compared to the optimal solution for each group using at most $\lceil k/g \rceil$ targets. \Cref{ex:lowerbound}, on the other hand, shows an instance where no solutions with $> 1/g$ simultaneous approximation for the groups is possible for either of the benchmarks. Therefore, there is a gap of $\mathcal{O}(g^2)$ for the first, and a gap of $\mathcal{O}(g)$ for the second benchmark. Finding the optimal order of approximation guarantees for these benchmarks and tight lower bounds are the main problems left open by our work.

%% file: appendix-missing-proof.tex
\section{Missing Proofs of \Cref{sec:max_total_improvement}}\label{app:missing_max_tot}
\subsection{Proof of~\Cref{thm:total-improvement}}

\begin{numberedtheorem}{\ref{thm:total-improvement}}
\Cref{def:recurrence-dp-one-group} finds a set of targets that achieves the optimal social welfare (maximum total improvement) that is feasible using at most $k$ targets given $n$ agents. The algorithm runs in $\mathcal{O}(n^3)$.
\end{numberedtheorem}

\begin{proof}
Proof of correctness follows by induction. Suppose that the value computed for all $T(\tau',\kappa')$ where $(\tau',\kappa') < (\tau,\kappa)$ is correct. Here ``$<$'' means 
$(\tau',\kappa')$ is computed before $(\tau,\kappa)$ which is when $\kappa'<\kappa$ and $\tau'\geq \tau$.  First, if either $\tau= \tau_{\text{max}}$ or $\kappa= 0$, the induction hypothesis holds since $T(\tau_{\max},\kappa)=0$ for all $1\leq \kappa\leq k$, and $T(\tau,0) = 0$, for all $\tau\in\T_p$. To show the inductive step holds note that the algorithm considers the optimal value for $T(\tau,\kappa)$ as the maximum of the $\displaystyle T(\tau',\kappa-1) \ + \textstyle\sum_{\tau\leq p_i <\tau'\text{ s.t. } \tau'-p_i \leq\Delta_i}(\tau'-p_i)$ over all the possible placement of the leftmost target $\tau'$. Since $T(\tau',\kappa-1)$ is computed correctly by the induction hypothesis and all the possible placements of the leftmost target are considered, the value obtained at $T(\tau,\kappa)$ is optimal and correct. 

Now we proceed to bounding the time-complexity. There are $\mathcal{O}(nk)$ subproblems to be computed. Consider a 
pre-computation stage where  $\sum_{\tau\leq p_i <\tau'\text{ s.t. } \tau'-p_i \leq\Delta_i}(\tau'-p_i)$ is computed for all pairs of $\tau, \tau'\in \T_p$. This stage takes $\mathcal{O}(n^3)$ time. Computation of each subproblem $T(\tau,\kappa)$ for all $\tau\in \T_p$ and $1\leq \kappa \leq k$ requires $\mathcal{O}(n)$ operations. This is because to compute $\max$ in property 3), we compute $\displaystyle T(\tau',\kappa-1) \ + \textstyle\sum_{\tau\leq p_i <\tau'\text{ s.t. } \tau'-p_i \leq\Delta_i}(\tau'-p_i)$ for $\mathcal{O}(n)$ potential target levels greater than $\tau$, for which each takes $\mathcal{O}(1)$ time. Since there are $\mathcal{O}(nk)$ subproblems, the running time of the algorithm is $\mathcal{O}(n^2k +n^3) = \mathcal{O}(n^3)$. 
\end{proof}

\section{Missing Proofs of \Cref{sec:max_min}}\label{app:missing_max_min}

\subsection{Proof of~\Cref{prop:running-time-fairness-exact}}
\begin{numberedtheorem}{\ref{prop:running-time-fairness-exact}}
\Cref{recurrence-exact-fairness-objective} constructs the Pareto frontier for groups' social welfare using at most $k$ targets given $n$ agents in $g$ groups, and has a running time of $\mathcal{O}(n^{g+2}kg\Delta_{\max}^g)$, where $\Delta_{\max}$ is the maximum improvement capacity.
\end{numberedtheorem}

\begin{proof}
Proof of correctness follows by induction and it is along the same lines as proof of~\Cref{def:recurrence-dp-one-group}. 
Suppose that Pareto-frontiers constructed for all $T(\tau',\kappa')$ where $(\tau',\kappa') < (\tau,\kappa)$ is correct. Here ``$<$'' means 
$(\tau',\kappa')$ is computed before $(\tau,\kappa)$ which is when $\kappa'<\kappa$ and $\tau'\geq \tau$.  First, if either $\tau= \tau_{\text{max}}$ or $\kappa= 0$, the induction hypothesis holds since $T(\tau_{\max},\kappa)=\emptyset$ for all $1\leq \kappa\leq k$, and $T(\tau,0) = \emptyset$, for all $\tau\in\T_p$. The inductive step holds since the algorithm considers all the possible placement of the leftmost target $\tau'$. Since $T(\tau',\kappa-1)$ is computed correctly by the induction hypothesis and all the possible placements of the leftmost target are considered, the Pareto-frontier constructed at $T(\tau,\kappa)$ is correct.

Now we proceed to bounding the time complexity. Initially, in a pre-computation stage, for each pair of targets $\tau, \tau'\in \mathcal{T}_p$,  $\sum_{\tau\leq p_i <\tau'\text{ s.t. } \tau'-p_i \leq\Delta_{\ell}} \mathbbm{1}\big\{i\in G_{\ell}\big\}(\tau'-p_i)$ is pre-computed for all groups and is stored in a tuple of size $g$. This stage can be done in $\mathcal{O}(n^3)$.
Each set $T(\tau,\kappa)$ has size at most $(n\Delta_{\max}+1)^g$, since each individual can move for one of the values $\{0,\cdots,\Delta_{\max}\}$ and therefore, the total improvement in each group is 
one of the values $\{0,\cdots,n\Delta_{\max}\}$.
At each step of the recurrence, given the information stored in the pre-computation stage, the summation can be computed in $\mathcal{O}(g)$. 
When computing a subproblem $T(\tau,\kappa)$, the recurrence searches over $\mathcal{O}(n)$ targets $\tau'\in \T_p$, and at most $(n\Delta_{\max}+1)^g$ tuples of group improvement in $T(\tau',\kappa-1)$. As a result, solving each subproblem takes $\mathcal{O}(ng(n\Delta_{\max})^g)$.
The total number of subproblems that need to get solved is $\mathcal{O}(nk)$.
Therefore, the total running time of the algorithm is $\mathcal{O}(n^{g+2}kg\Delta_{\max}^g+n^3)$ = $\mathcal{O}(n^{g+2}kg\Delta_{\max}^g)$.
\end{proof}

\subsection{Proof of~\Cref{cor:exact-fairness}}

\begin{numberedcorollary}{\ref{cor:exact-fairness}}
There is an efficient algorithm that finds a set of at most $k$ targets that maximizes minimum improvement across all groups, i.e., maximizing $\min_{1 \leq \ell \leq g} \sw_\ell$.
\end{numberedcorollary}
\begin{proof}
\Cref{recurrence-exact-fairness-objective} constructs the Pareto frontier for groups' social welfare. By iterating through all Pareto-optimal solutions, we can find the solution that maximizes the minimum improvement across all groups. There are at most $(n\Delta_{\max}+1)^g$ Pareto-optimal solutions. Finding the minimum improvement in each solution takes $\mathcal{O}(g)$. Therefore, in total, finding the solution that maximizes the minimum improvement across all groups takes $\mathcal{O}(g(n\Delta_{\max})^g)$.
\end{proof}

%% file: FPTAS-maxmin.tex
\section{An FPTAS for Maximizing Minimum Group Improvement}
\label{sec:appendix-FPTAS}

In this section, we present a Fully Polynomial Time Approximation Scheme (FPTAS) to maximize minimum improvement across all groups. Here, we assume that each group $\ell$ has its own improvement capacity $\Delta_{\ell}$.

\begin{algo}
\label{algo:Max-min}
The algorithm considers two separate cases of $k< g$, and $k\geq g$. For the $k\geq g$ case, the algorithm finds a set of $k$ targets that approximates the max-min objective within a factor of $1-\eps$ for any arbitrary value of $\eps>0$. For the $k<g$ case, it finds an optimal solution for the max-min objective.

For the $k\geq g$ case, there exists an FPTAS for the max-min objective as follows. First, run a dynamic program using the following recursive function to get a set of Pareto-optimal solutions. In this Pareto-frontier, we show the solution that maximizes minimum improvement across all groups, gives a ($1-\eps$)-approximation for the max-min objective. In the recurrence, $\mu_{\ell} = \varepsilon\Delta_{\ell}/(16kg^3)$ for $1\leq \ell\leq g$, and $\Delta_{\ell}$ is the improvement capacity of agents in group $\ell$.

\begin{gather*} 
\begin{align*}
&\mathcal{F}(\tau',k') = \Bigg\{\Bigg(\mu_{\ell}\floor*{\frac{I'_{\ell}+\Big(\sum_{\substack{\tau'\leq p_i <\tau\\\text{ s.t. } \tau-p_i \leq\Delta_{\ell}}} \mathbbm{1}\Big\{i\in G_{\ell}\Big\}(\tau-p_i)\Big)}{\mu_{\ell}}}\Bigg)_{\ell=1}^{g},  &\text{ s.t. } (I'_{\ell})_{\ell=1}^{g}\in \mathcal{F}(\tau,k'-1), \tau\in \mathcal{T}_p, \tau\geq \tau' \Bigg\}
\end{align*}
\end{gather*}

Intuitively, $\mathcal{F}(\tau',k')$ stores the rounded down values of the feasible tuples of group improvements when all agents on or to the right of $\tau'$ are available and $k'$ targets are used. The corresponding set of targets used to construct the improvement tuples in $\mathcal{F}(\tau',k')$ is kept in a hash table $\mathcal{S}(\tau',k')$, whose keys are the improvement tuples in $\mathcal{F}(\tau',k')$. The dynamic program ends after computing $\mathcal{F}(\tau_{\min},k)$ and $\mathcal{S}(\tau_{\min},k)$. At the end, we output the set of targets in $\mathcal{S}(\tau_{\min},k)$ 
corresponding to the improvement tuple that maximizes the improvement of the worst-off group.
~\Cref{thm:approx-guarantee} shows that this algorithm gives a ($1-\eps$)-approximation for the max-min objective when $k\geq g$.

When $k<g$, for each subset of $\T_p$ of size at most $k$ that is corresponding to a placement of targets, we store its corresponding improvement tuple. Next, we iterate through all improvement tuples and output the one that maximizes minimum improvement.
\end{algo}

\begin{lemma}
\label{thm:approx-guarantee} 
\Cref{algo:Max-min} gives a ($1-\eps$)-approximation for the max-min objective when $k\geq g$.
\end{lemma}

\begin{proof}
The proof is by induction. Consider an improvement tuple $(I_1,\cdots,I_g)$ corresponding to an arbitrary set of $k-1$ targets, 
and let $(I'_1,\cdots,I'_g)$ denote the rounded down values where $I'_{\ell} = \mu_{\ell}\floor{\frac{I_{\ell}}{\mu_{\ell}}}$ for all $1\leq \ell \leq g$. Suppose that for all $1\leq \ell \leq g$,
$I'_{\ell} \geq I_{\ell}-(k-1)\mu_{\ell}$.

Now consider an improvement tuple $(J_1,\cdots,J_g)$ corresponding to an
arbitrary set of $k$ targets. For each $1\leq \ell \leq g$, let $J'_{\ell} = \mu_{\ell}\floor{\frac{J_{\ell}}{\mu_{\ell}}}$. We show that for each $1\leq \ell\leq g$, $J'_{\ell} \geq J_{\ell} - k\mu_{\ell}$.
For all $1\leq \ell \leq g$, let $J_{\ell} = L_{\ell} + I_{\ell}$, where $L_{\ell}$ is the improvement of group $\ell$ that the leftmost target provides, and $I_{\ell}$ captures the true improvement of group $\ell$ that the remaining $k-1$ targets provide. Let $I'_{\ell} = \mu_{\ell}\floor{\frac{I_{\ell}}{\mu_{\ell}}}$. 
Then $J'_{\ell} = \mu_{\ell}\floor{\frac{L_{\ell}+I'_{\ell}}{\mu_{\ell}}}$ implying that $J'_{\ell}\geq L_{\ell} + I'_{\ell}-\mu_{\ell}$. By the induction hypothesis, $I'_{\ell} \geq I_{\ell}-(k-1)\mu_{\ell}$. Therefore,

\[
J'_{\ell} \geq L_{\ell} + I'_{\ell}-\mu_{\ell}
\geq L_{\ell}+I_{\ell}-(k-1)\mu_{\ell}-\mu_{\ell}
= L_{\ell}+I_{\ell}-k\mu_{\ell}
\]
Therefore, for each set of $k$ targets, the rounded improvement of each group ${\ell}$ stored in the table is within an additive factor of $k\mu_{\ell} = \eps\Delta_{\ell}/(16g^3)$ of its true improvement. We argue that in the solution returned by the algorithm, improvement of each group is at least $(1-\eps)OPT$. First, when $k\geq 1$, each group can improve for at least $\Delta_{\ell}$ by setting a target within a distance of $\Delta_{\ell}$ from its rightmost agent. Now, using~\Cref{thm:approx} when $k\geq g$, there exists a solution that is simultaneously $1/(16g^3)$-optimal for all groups. Therefore, the optimum value of the max-min objective is at least $OPT\geq \Delta_{\ell}/(16g^3)$ for all $1\leq \ell \leq g$. Therefore, for each solution consisting of $k$ targets, the rounded improvement of each group is within an additive factor of $\eps OPT$ of its true improvement.
As a result, the minimum group improvement in the returned solution is at least $(1-\eps)OPT$. \end{proof}

In the following, we bound the approximation factor of our algorithm in both cases of $k\geq g$ and $k<g$.

\begin{corollary}
\Cref{algo:Max-min} described above gives a ($1-\varepsilon$)-approximation for the max-min objective.
\end{corollary}

\begin{proof}
For the case of $k\geq g$, by~\Cref{thm:approx-guarantee} the algorithm outputs a ($1-\varepsilon$)-approximation. For $k<g$, it outputs an optimum solution. Therefore, in total, it gives a ($1-\varepsilon$)-approximation for the max-min objective.
\end{proof}

In the following, we bound the time-complexity of the algorithm.

\begin{theorem}
\label{thm:running-time}
\Cref{algo:Max-min} has a running time of
$\mathcal{O}(n^{g+2}k^{g+1}g^{3g+1}/\eps^g)$.
\end{theorem}

\begin{proof}
Initially, in a pre-computation stage, for each pair of targets $\tau, \tau'\in \mathcal{T}_p$,  $\sum_{\tau'\leq p_i <\tau\text{ s.t. } \tau-p_i \leq\Delta_{\ell}} \mathbbm{1}\Big\{i\in G_{\ell}\Big\}(\tau-p_i)$ is pre-computed for all groups and is stored in a tuple of size $g$.
This stage can be done in $\mathcal{O}(n^3)$.
Now, first consider the case where $k\geq g$. We show the dynamic programming algorithm using recurrence $\mathcal{F}(\tau',k')$ has a running time of 
$\mathcal{O}(n^{g+2}k^{g+1}g^{3g+1}/\eps^g)$.
 Each set $\mathcal{F}(\tau',k')$ and $\mathcal{S}(\tau',k')$ has size at most $\prod_{\ell=1}^{g}(n\Delta_{\ell}/\mu_{\ell})^{g} = (16nkg^3/\eps)^{g}$.
At each step of the recurrence, given the information stored in the pre-computation stage, the summation can be computed in $\mathcal{O}(g)$ .
When computing $\mathcal{F}(\tau',k')$, the recurrence searches over $\mathcal{O}(n)$ targets $\tau\in \mathcal{T}_p$, and at most  $\prod_{\ell=1}^{g}(n\Delta_{\ell}/\mu_{\ell})^{g} = (16nkg^3/\eps)^{g}$ tuples of group improvement in $\mathcal{F}(\tau,k'-1)$. As a result, solving each subproblem takes $\mathcal{O}(ng(nkg^3/\eps)^{g})$. The total number of subproblems that need to get solved is $\mathcal{O}(nk)$. Therefore, the total running time of computing $\mathcal{F}(\tau_{\min},k)$ is 
$\mathcal{O}(n^{g+2}k^{g+1}g^{3g+1}/\eps^g)$.

Next, consider the case where $k<g$. The algorithm considers $\mathcal{O}(n^g)$ placements of targets. Given the pre-computation stage,
computing the improvement tuple corresponding to each placement of targets   takes $\mathcal{O}(kg)$. As a result, this case takes $\mathcal{O}(kgn^g)$. 

Therefore, the total running time of algorithm is 
$\mathcal{O}(n^3+n^{g+2}k^{g+1}g^{3g+1}/\eps^g+kgn^g) = \mathcal{O}(n^{g+2}k^{g+1}g^{3g+1}/\eps^g)$.
\end{proof}

%% file: appendix_missing_approximation.tex
\section{Missing Proofs of \Cref{sec:approx_optimality}}\label{app:approximation}

\begin{numberedlemma}{\ref{lm:make_distant}}
Consider solution $\T: \tau_1< \tau_2< \ldots$ with total improvement $I$ such that for all $j$,  $\tau_{j+2}-\tau_j \geq \Delta$. Consider the procedure in \Cref{def:distant_procedure}. This procedure results in a solution $\T': \tau'_1< \tau'_2< \ldots$ where $\forall j \; \tau'_{j+1}-\tau'_j \geq 2\Delta$, has total improvement at least $ I/4$, and $|\T'| \leq \lceil |\T|/4 \rceil$. Particularly, for $|\T|\leq \lceil k/g \rceil$ where $k \geq g$, the number of final targets, $|\T'|$, is at most $\lfloor k/g \rfloor$.
\end{numberedlemma}

\begin{proof}[Proof of \Cref{lm:make_distant}]
Since the best out of $4$ parts have been selected, the total improvement at the end of the procedure is at least $1/4$ fraction of $I$.
In addition, in the final set, every pair of consecutive targets are indexed $\tau_{j}$ and $\tau_{j+4}$. Therefore, since originally for all $j$, $\tau_{j+2}-\tau_j \geq \Delta$, we have $\tau_{j+4}-\tau_j \geq 2\Delta$. Finally, since in each set of $\tau_j,\ldots, \tau_{j+4}$ exactly one target is selected, the final number of targets is at most $\lceil |\T|/4 \rceil$.
\end{proof}

\begin{numberedlemma}{\ref{lm:step_three}}
At the end of step $3$ in \Cref{alg:approx}, (i) the distance between every two targets in $\T_\ell$ is at least $\Delta$; (ii) each target $\tau \in \T_\ell$ is optimal, i.e., maximizes total improvement for agents in $G_\ell \cap [\tau-\Delta, \tau)$; and (iii) the total amount of improvement of $G_\ell$ using solution $\T_\ell$ does not decrease compared to the previous step.
\end{numberedlemma}

\begin{proof}[Proof of \Cref{lm:step_three}]
Let $\tau$ be a target at the beginning of step $3$ and $\tau'$ be its replacement at the end of this step. 

We first prove statement (i). First, we argue for agents in $[\tau-\Delta, \tau)$, the optimal target $\tau'$  belongs to $[\tau, \tau+\Delta]$. Intuitively, the reason is that all these agents afford to improve to $\tau$; therefore, a target smaller than $\tau$ is suboptimal. Also, none of the agents affords to improve beyond $\tau+\Delta$. More formally, if $\tau'<\tau$, agents in $[\tau-\Delta, \tau')$ improve less compared to a target at $\tau$ and agents in $[\tau',\tau)$ do not improve. On the other hand, if $\tau'>\tau+\Delta$, none of the agents can reach $\tau'$ and the total improvement for these agents will be $0$. Therefore, at the end of this step, every target $\tau$ is replaced with $\tau' \in [\tau, \tau+\Delta]$. Now, by \Cref{obs:delta_apart} and \Cref{lm:make_distant}, the distance between consecutive targets at the end of step $2$ is at least $2\Delta$. Therefore, after the modification explained (shifting each target to the right by less than $\Delta$) this distance decreases by at most $\Delta$ and becomes at least $\Delta$. 

Now, we move on to statement (ii). We need to argue if $\tau'$ is optimal for agents in $[\tau-\Delta, \tau)$, it is also optimal for agents in $[\tau'-\Delta, \tau')$. By \Cref{lm:make_distant}, at the beginning of step $3$, there are no targets in $(\tau, \tau+2\Delta)$; more specifically, there are no targets for agents in $[\tau,\tau+\Delta)$ and these agents get eliminated in this step. Therefore, since $\tau'$ belongs to $[\tau, \tau+\Delta]$, as shown in the proof of statement (i), we only need to argue that if $\tau'$ is optimal for $[\tau-\Delta, \tau)$, it is also optimal for $[\tau'-\Delta, \tau)$. Suppose this was not the case, and there was another target $\tau''$ which was optimal for this set. Since the agents in $[\tau'-\Delta, \tau)$ are the only agents with positive amount of improvement for target $\tau'$, replacing $\tau'$ with $\tau''$ would result in higher improvement for the whole set of agents in $[\tau-\Delta, \tau)$ which is in contradiction with definition of $\tau'$.  

Finally, we argue statement (iii). In step $3$, the agents not improving in step $2$ have been eliminated and the new targets only (weakly) increased the total improvement of the remaining agents. Therefore, the total amount of improvement does not decrease in this step.
\end{proof}

\begin{numberedlemma}{\ref{lm:step_four}}
Consider $\T$ as the union of all solutions at the end of step $3$. For all $\tau \in \T$, consider the interval $[\tau-\Delta, \tau)$ which consists of agents that improve to target $\tau$ if it were the only target available. At the end of step $4$, (i) there will be a target in $[\tau-\Delta+\Delta/g, \tau]$, and (ii) there will be no targets in $(\tau-\Delta, \tau-\Delta+\Delta/g)$.
\end{numberedlemma}
\begin{proof}[Proof of \Cref{lm:step_four}] Statement (i) is equivalent to (i') for any $s \in S$, there will be a target in $[s+\Delta/g, s+\Delta)$; and statement (ii) is equivalent to (ii') for any $s \in S$, there will be no targets in $(s, s+ \Delta/g)$. We prove (i') and (ii').

We first show the size of each part is at most $g$; i.e. $\forall i, |S_i| \geq g$. The proof is by contradiction. Suppose there exists $|S_i| \geq g+1$. Therefore, there exist $s_j < s_{j'} \in S_i$ and group index $\ell$, such that $s_j+\Delta, s_{j'}+\Delta \in \T_\ell$, and all $s$ satisfying $s_j< s< s_{j'} \in S_i$ corresponding to targets in distinct groups other than $\ell$. Therefore, there are at most $g-1$ such $s$. Hence, $s_{j'}-s_j < g\times \Delta/g = \Delta$, implying there are two targets in $\T_\ell$ at distance strictly less than $\Delta$ which is in contradiction with \Cref{lm:step_three}.

Now, we prove statement (i''). In step $4$, the final target corresponding to part $S_i : s_u\leq s_{u+1}\leq \ldots \leq s_v$ is defined as $\tau^*_i = \min\{\tau_v, s_{v+1}\}$. By definition, $\tau^*_i \leq s_{v+1}$; therefore, it is (weakly) to the left of any $s_j$ for $j \geq v+1$. Also,  using $|S_i| \leq g$, $s_v < s_u+ (g-1)\Delta/g$, which implies $\tau_u - s_v > \Delta/g$,
and since by definition, $s_{v+1}-s_v \geq \Delta/g$, both $s_{v+1}$ and $\tau_u$ are at least at distance $\Delta/g$ to the right of $s_v$ and any $s_j$ such that $j \leq v$. This proves statement (i'').  

Finally, we prove (i'). In the proof of (ii'), we showed that $\tau^*_i \geq s_v + \Delta/g$ which implies $\tau^*_i \geq s + \Delta/g, \; \forall s \in S_i$. Therefore, it suffices to show $\tau^*_i \leq s_u + \Delta$, which then implies $\tau^*_i \leq s + \Delta, \; \forall s \in S_i$. The definition of $\tau^*_i$ directly implies $\tau^*_i \leq s_u + \Delta$.
\end{proof}

%% file: appendix_example.tex
\section{Distance between consecutive target levels}

\Cref{obs:delta_apart} shows it is without loss of optimality to assume the distance between every other targets is at least $\Delta$ in the common improvement capacity model. The following example investigates this property for \emph{consecutive} targets, and shows an instance where the distance between two consecutive targets is arbitrarily small compared to $\Delta$ in the optimal solution.

\begin{example}\label{ex:consecutive_less_than_delta}
Suppose $\Delta=1$ and there is no limit on the number of targets.  Suppose there is an agent at position $0$, an agent at position $1$, and $m$ agents at position $1+1/m$. The optimal solution is $\T = \{\tau_1=1, \tau_2=1+1/m, \tau_3=2+1/m\}$. As $m \rightarrow \infty$, the distance between $\tau_1$ and $\tau_2$ approaches $0$.
\end{example}